\documentclass[11pt]{amsart}
\usepackage[margin=1in]{geometry}
\usepackage{amsfonts, float}
\usepackage{amssymb}
\usepackage[dvips]{graphics}
\usepackage{epsfig}
\usepackage{euscript}
\usepackage{color}
\usepackage{cite}

\usepackage[all]{xy}
\usepackage{multirow}

\usepackage{fancyvrb}

\usepackage[pdftex]{hyperref}

 \newtheorem{thm}{Theorem}[section]

 \newtheorem{prop}[thm]{Proposition}
 
 \theoremstyle{definition}
 
 \theoremstyle{remark}

 \numberwithin{equation}{section}
 
 \def\idtyty{{\mathchoice {\mathrm{1\mskip-4mu l}} {\mathrm{1\mskip-4mu l}} %
{\mathrm{1\mskip-4.5mu l}} {\mathrm{1\mskip-5mu l}}}}

\newcommand{\bC}{{\mathbb C}}

\newcommand{\bZ}{{\mathbb Z}}
\newcommand{\cA}{{\mathcal A}}
\newcommand{\cB}{{\mathcal B}}
\newcommand{\cH}{{\mathcal H}}

\newcommand{\caA}{{\mathcal A}}

\newcommand{\caG}{{\mathcal G}}
\newcommand{\caH}{{\mathcal H}}

\newcommand{\bbC}{{\mathbb C}}

\newcommand{\bbN}{{\mathbb N}}

\newcommand{\bbZ}{{\mathbb Z}}

\newcommand{\str}{^{*}}

\newcommand{\Tr}{\mathrm{Tr}}

\newcommand{\braket}[2]{\left\langle #1 , #2\right\rangle}

\newcommand{\bra}[1]{\langle #1 \vert} 

\newcommand{\ket}[1]{\vert #1 \rangle}

\newcommand{\spec}{\mathrm{spec}}

\newcommand{\be}{\begin{equation}}
\newcommand{\ee}{\end{equation}}
\newcommand{\bea}{\begin{eqnarray}}
\newcommand{\eea}{\end{eqnarray}}
\newcommand{\beann}{\begin{eqnarray*}}
\newcommand{\eeann}{\end{eqnarray*}}

\newcommand{\eq}[1]{(\ref{#1})}

\newcommand{\Lamdif}[1]{\Lambda_{n+1}\backslash\Lambda_{#1}}

\newcommand{\mbf}[1]{\mathbf{#1}}
\newcommand{\bfN}{\mathbf{N}}

\begin{document}

\renewcommand{\thefootnote}{\fnsymbol{footnote}}
\title[$d$-Dimensional PVBS models]{Product Vacua and Boundary State Models in
$d$-Dimensions}

\author[S. Bachmann]{Sven Bachmann}
\address{Mathematisches Institut \\
Ludwig-Maximilians-Universit\"at M\"unchen \\
Theresienstr.~39 \\
D-80333 M\"unchen, Germany}
\email{sven.bachmann@math.lmu.de}

\author[E. Hamza]{Eman Hamza}
\address{Department of Physics\\ Faculty of Science \\ Cairo University \\ Cairo 12613, Egypt}
\email{emanhamza@sci.cu.edu.eg}

\author[B. Nachtergaele]{Bruno Nachtergaele}
\address{Department of Mathematics\\
University of California, Davis\\
Davis, CA 95616, USA}
\email{bxn@math.ucdavis.edu}

\author[A. Young]{Amanda Young}
\address{Department of Mathematics\\
University of California, Davis\\
Davis, CA 95616, USA}
\email{amyoung@math.ucdavis.edu}

\begin{abstract} 

We introduce and analyze a class of quantum spin models defined on $d$-dimensional
lattices $\Lambda \subseteq \bZ^d$, which we call {\em Product Vacua with Boundary States} (PVBS).
We characterize their ground state spaces on arbitrary finite volumes and study the thermodynamic limit. 
Using the martingale method, we prove that the models have a gapped excitation spectrum on $\bbZ^d$
except for critical values of the parameters. For special values of the parameters we show that the excitation 
spectrum is gapless. We demonstrate the sensitivity of the spectrum to the existence and orientation of 
boundaries. This sensitivity can be explained by the presence or absence of edge excitations. In particular, we 
study a PVBS models on a slanted half-plane and show that it has gapless edge states but a gapped excitation 
spectrum in the bulk.

\end{abstract}

\maketitle

\date{\today }

\footnotetext[1]{Copyright \copyright\ 2014 by the authors. This
paper may be reproduced, in its entirety, for non-commercial
purposes.}



\section{Introduction}\label{sec:Introduction}


Gapped ground state phases of quantum lattice systems are receiving much renewed
attention lately because they may support topological order, a feature of
interest as a potential path to realizing robust quantum memory
\cite{kitaev:2003}.
Particular effort has been devoted to classifying gapped ground state phases
\cite{Kitaev:2009,Chen:2011iq,Chen:2013,Schuch:2011,Bachmann:2011kw,Hastings:2013wp,Duivenvoorden:2013,Haegeman:2013,Bachmann:2014a}.

In this work we generalize the class of models introduced in
\cite{Bachmann:2012bf}, called {\em Product Vacua with Boundary States} (PVBS),
to higher dimensions and provide a proof for a subclass of these models that
they have a non-vanishing gap above the ground state. These models are quantum spin systems
with nearest neighbor interactions defined on $\bZ^d$. The bulk ground state is
a simple product state (the vacuum) but there is also an interesting
structure of edge states. These models are defined in Section \ref{sec:SetupAndResults}.

PVBS models have in common with many other examples of multidimensional models
with known or expected gapped ground states  that they are frustration-free and
that they can be considered as particularly simple representatives of a `gapped
ground state phase'. Toy models of this type are sometimes referred to in the
literature as Renormalization Group Fixed Points (RGFP).
Well-known examples are the Toric Code model \cite{kitaev:2003} and the models
with Product of Entangled Pairs  (PEPS) ground states studied in
\cite{Perez-Garcia:2010,Cirac:2013, Schuch:2013,Yang:2014}. In contrast to most
of the previously studied models for which a non-vanishing spectral gap is
proved, the interaction terms of the PVBS models are not mutually commuting and
the excitations have a physical, in particular non-flat, dispersion relation.

The focus of the present paper is to prove the non-vanishing spectral gap and
study the nature of the edge states. Section~\ref{sec:GroundStateSpace} carries
out a detailed discussion of the finite volume ground state spaces and of their
limits to various infinite volumes. In particular, we show that the ground state
degeneracy depends on the geometry of the infinite volume and on the presence of
edges or corners. Obtaining rigorous lower bounds for the spectral gap of
multidimensional models with non-commuting interaction terms has proved to be
notoriously difficult. In this work we prove such lower bounds for PVBS models
with one species of particles in $d$-dimensions on infinite volumes that are
obtain as limits of rectangular boxes. We use the martingale method from
\cite{nachtergaele:1996} to accomplish this in Section~\ref{sec:LowerBound}.
In Section \ref{sec:UpperBound} we show that PVBS models defined on specific
infinite volumes, such as the half-infinite plane with a 45 degree boundary,
that have the exact the same interactions as the PVBS models that we prove are
gapped in Section~\ref{sec:LowerBound}, can have a gapless excitation spectrum,
meaning the GNS Hamiltonian has gapless spectrum. It is then also clear that, in
both finite and infinite volumes, the size of the gap above the ground state(s)
depends on the shape of the domain and, specifically, on the orientations of its
edges.

It is important that the ground state and spectral gap for the PVBS models are
stable under small perturbations of the interactions. The essential ingredient
in proving stability for frustration-free models
in~\cite{Michalakis:2013} is the so-called `Local Topological Quantum Order'
condition, which we prove holds for the ground states of the PVBS models in
Section~\ref{sec:Discussion}. We conclude the paper by clarifying the
fundamental difference between two possible ways of closing the gap in PVBS
models. The first, the vanishing of the gap at a critical value of the coupling
(here $\lambda=1$) is a bulk phenomenon. The second involves the geometry of the
boundary of the system and the gap closes due to edge excitations. As we
explain, in the second case the `bulk gap' stays open.


\section{General Setup and Results}\label{sec:SetupAndResults}


\subsection{The $d$-Dimensional PVBS Model}
The class of quantum spin chains introduced in \cite{Bachmann:2012bf}, called
Product Vacua with Boundary States (PVBS), has an $(n+1)$- dimensional state space
at each site for which $n\geq 1$ denotes the number of particle species in the
model, labeled  $i=1,\ldots, n$. At most one particle may occupy a given site at
any time. If a particle of species $i$ occupies site $x$, the state at site $x$
is denoted by $\ket{i}$, and if no particle occupies the site, then the
state is denoted by $\ket{0}$. The Hamiltonian for a chain of $L$ spins is given
by 
$$ H_{[1,L]}=\sum_{x=1}^{L-1} h_{x,x+1}, $$ 
where $h_{x,x+1}$ acts non-trivially on the states of the nearest neighbor pair
at sites $x$ and $x+1$ by a copy of the interaction term $h$. The interaction
$h$ contains both hopping terms for the particles and a repulsive interaction
between two particles of the same species occupying neighboring sites.  
Specifically, $h$ is of the form
\be
\label{pvbs_int} h=\sum_{i=1}^n |\hat{\phi}_i\rangle\langle\hat{\phi}_i| +
\sum_{1\leq i \leq j\leq n}^n |\hat{\phi}_{ij}\rangle\langle\hat{\phi}_{ij}|,
\ee 
where the vectors $\phi_{ij}\in \mathbb{C}^{n+1}\otimes\mathbb{C}^{n+1}$ are
given by 
$$ \phi_{i}=\ket{0,i}-\lambda_i\ket{i,0} \label{phii},\quad
\phi_{ij}=\lambda_i \ket{i,j}-\lambda_j\ket{j,i}
\label{phij},\quad \phi_{ii}=\ket{i,i} $$ 
for $i=1,\ldots, n$ and $i\neq j=1,\ldots,n$. The parameters satisfy $
\lambda_i\geq 0$, for $0\leq i\leq n$, and we put $\lambda_0=1$ for later convenience. 
The notation $\hat{\phi}$ denotes the normalized vector $\Vert \phi\Vert^{-1}\phi$. 

We are interested in defining higher-dimensional PVBS models. To do this we
consider finite connected subsets $\Lambda \subseteq \bZ^d$ as oriented graphs
with edges of the form $(\mbf{x}, \mbf{x}+e_k)$,  where $e_k$, $1 \leq k \leq
d$, are the canonical basis vectors of $\bbZ^d$.
For each oriented edge $(\mbf{x}, \mbf{x}+e_k) $ of a graph $\Lambda$, we then
consider an interaction term of the same form as that of the one-dimensional
model defined in \eq{pvbs_int}:
\be\label{pvbs_ham2}
H_\Lambda= \sum_{\mbf{x}\in\Lambda} \sum_{\substack{k=1, \ldots d \\ \mbf{x}+e_k \in\Lambda}} h^{(k)}_{\mbf{x},\mbf{x}+e_k} 
\ee 
Here, for each direction $k=1,\ldots d$, we may choose different values for the
parameter, say $\lambda_{(i,k)}$, appearing in the PVBS interaction \eq{pvbs_int}. 

As in the one-dimensional case, the ground states space of $H_\Lambda$ is equal
to its kernel and can be explicitly constructed. If $\Lambda$ is connected, has
at least $n$ vertices, and $\lambda_{(i, k)}\in (0,\infty)$ for all $i=1, 
\ldots,n$ and $k=1,\ldots, d$, then the dimension of the ground state space is $2^n$. Under these assumptions, the ground state space
has a basis $\{\psi_{M}^\Lambda\}$ labeled by the subsets
$M\subset\{1,\dots, n\}$.
For $M=\emptyset$, we have 
$$ \psi_\emptyset^\Lambda= \bigotimes_{x\in\Lambda}
\ket{0}. 
$$ 
For $M=\{i_1,\ldots, i_m\}$, the corresponding ground state vector
is given by
\begin{align}
\psi_{M}^\Lambda& = \sum_{\substack{\mbf y_1,\ldots,\mbf y_m\in\Lambda \\ \mbf y_j\neq \mbf y_i \text{ if }i\neq j}}
\prod_{j=1}^m\prod_{k=1}^d
\lambda_{(i_j,k)}^{y_{j_k}} \xi(M,\mbf y_1,\ldots,\mbf y_m)\label{gsvectors} \\
\mbf y_j&=(y_{j_1},\ldots,y_{j_d}),\nonumber
\end{align}
where $\mbf y_1,\ldots,\mbf y_m\in\Lambda$, all $\mbf y_j$
distinct, $\xi(M,\mbf y_1,\ldots,\mbf y_m)$ denotes the unit vector with a particle of
type $i_j$ in position $\mbf y_j$. It is straightforward to check that
$h^{(k)}_{\mbf{x},\mbf{x}+e_k}\psi_{M}^\Lambda=0$, for all pairs
$\{\mbf x,\mbf x+e_k\}\subset \Lambda$. Therefore, the vectors (\ref{gsvectors}) are $2^n$
mutually orthogonal zero-energy ground states of (\ref{pvbs_ham2}).
More details on the ground state space for the one-species case are given in
Section \ref{sec:GroundStateSpace}.

The scope of this paper will be to describe various properties of the
$d$-dimensional PVBS model with a single species of particle. For the single
species model, the local Hilbert space at each site $\mbf{x}$ is then $\bC^2$
and we can simplify the notation and the definition of the Hamiltonian as
follows.
We abbreviate $\lambda_{(1, k)}$ by $\lambda_k$, and write the Hamiltonian on
any finite $\Lambda\subseteq \bZ^d$ as

\begin{equation} \label{eqn:PVBSHam}
H_{\Lambda} = \sum_{k=1}^d\sum_{\mbf{x}, \mbf{x}+e_{k} \in \Lambda}
h_{\mbf{x}, \mbf{x}+e_k}^{(k)}
\end{equation}
where $h_{\mbf{x}, \mbf{x} + e_k}^{(k)}$ acts non-trivially on the states associated
with the sites $\mbf{x}$, $\mbf{x}+e_k$ and is defined by
\begin{align} \label{eqn:PVBSInteraction}
h_{\mbf{x}, \mbf{x}+ e_k}^{(k)} & = \ket{1,1}\bra{1,1} +
\ket{\hat{\phi}_k}\bra{\hat{\phi}_k}, \\
\hat{\phi}_k & = \frac{1}{\sqrt{1 +
\lambda_k^2}}\left(\ket{0,1} - \lambda_k\ket{1, 0} \right).
\end{align}
The Hamiltonian $H_\Lambda$ preserves particle number, a fact that will be
useful throughout the paper.

We define the following orthonormal basis for the Hilbert space of the
one-species model, $\cH_{\Lambda}$.
Let $X \subseteq \Lambda$ where $\Lambda$ is finite. We define
$\xi_X^\Lambda \in \cH_{\Lambda}$ to be the product state vector with
exactly $\vert X \vert$ particles occupying the sites $\mbf{x} \in
X$, namely
\begin{equation}
\xi_{X}^{\Lambda}=\bigotimes_{\mbf{x}\in
\Lambda}\ket{\xi_{X}^{\Lambda}(\mbf{x})} \mbox{ with }\xi_{X}^{\Lambda}(\mbf{x})
= \begin{cases} 1& \mbf{x}\in X \\
0 & \mbf{x} \in \Lambda\setminus X
\end{cases}
\label{gs1}\end{equation}

We refer to the basis $\{\xi_{X}^{\Lambda} : X\subseteq \Lambda\}$ as the
canonical orthonormal basis $\cB_\Lambda$ of $\cH_\Lambda$. Note that contrary to the one-dimensional PVBS models discussed above, there is only one particle species here and the label $X$ refers to their positions.

In the case of a single-species PVBS model defined on a finite connected graph
$\Lambda \subseteq \bZ^d$ with $\lambda_k \in (0, \infty)$ for $k = 1, \ldots,
\, d$, the ground state space is two-dimensional. The normalized ground
state space vectors are given by
\begin{equation}\label{eqn:GSBasis}
\psi_0^{\Lambda} = \xi_{\emptyset}^\Lambda = \bigotimes_{\mbf{x} \in \Lambda}
\ket{0}, \hspace{.25 in}
\psi_1^{\Lambda} = \frac{1}{\sqrt{C_{\Lambda}}}\sum_{\mbf{x} \in
\Lambda}\lambda^{\mbf{x}}
\xi_{\{\mbf{x}\}}^{\Lambda}
\end{equation}
where 
\begin{equation}\label{CLambda}
C(\Lambda) = \sum_{\mbf{x} \in \Lambda}\lambda^{2\mbf{x}},
\end{equation}
with, for any $\mbf{x} = (x_1, \,x_2, \, \ldots, \, x_d)\in\mathbb{Z}^d$,
$\lambda^{\mbf{x}} =\prod_{k=1}^d\lambda_k^{x_k}$. Note that the subscript $0,1$ refers here to whether or not the ground state has a particle.

We briefly comment on the assumption $\lambda_k > 0$. First, if some but not all
of the $\lambda_k =0$, then there exist finite connected $\Lambda \subseteq
\bZ^d$ on which the model has other ground states in addition to the ones given
by (\ref{gs1}).
These additional ground states correspond to `stuck particles' and do not
appear to be of particular interest to us. Second, the model with complex parameters
$\lambda_k\in \mathbb{C}$ is unitarily equivalent to the model with each
parameter replaced by its absolute value. Therefore, we can restrict our attention to the
case of $\lambda_k>0$, for all $k=1,\ldots, d$.

\subsection{Results} 

We extend some of the results for the one-dimensional PVBS models, in particular the lower and 
upper bounds for the spectral gap above the ground states, to the $d$-dimensional models.
We also describe some new phenomena related to gapless edges states that only occur for $d>1$.

We first prove that the ground state space is, in fact, two-dimensional and spanned
by the orthonormal states given in equation \eqref{eqn:GSBasis}.
We then discuss the thermodynamic limit taken over different sequences of finite volumes. 
The nature of the ground states on infinite systems and the spectrum of low-lying excitations 
depends on the choice of the infinite volume, in particular on the presence and the nature
of boundaries. Our main results are the following.

\begin{prop}[Ground States on Finite Volumes]\label{prop:FiniteGSS}
For the one-species PVBS model defined on a finite, connected subset
$\Lambda \subseteq \bZ^d$ as given in \eqref{eqn:PVBSHam}, and
\eqref{eqn:PVBSInteraction} with $\lambda_k \in (0, \infty)$ for all
$k= 1, \ldots, d$, the ground state space $\caG_{\Lambda}$ is given by the
kernel of $H_\Lambda$, is two dimensional, and spanned by
\begin{equation}\label{eqn:FiniteGSB}
\psi_0^{\Lambda} = \bigotimes_{\mbf{x} \in \Lambda} \ket{0}, \hspace{.25
in} 
\psi_1^{\Lambda} = \frac{1}{\sqrt{C(\Lambda)}}\sum_{\mbf{x} \in
\Lambda}\lambda^{\mbf{x}}
\xi_{\{\mbf{x}\}}^{\Lambda}
\end{equation}
where $C(\Lambda)$ is the normalization constant defined in \eq{CLambda}.
\end{prop}

We are interested in the ground states defined for the system on infinite
subsets $\Gamma\subseteq \bZ^d$, that can be obtained as weak limits of the
finite volume ground states associated with the vectors $\psi_0^{\Lambda}$ and $
\psi_1^{\Lambda}$ of \eq{eqn:FiniteGSB}.
Such limits are taken over increasing sequences of finite subsets $\Lambda_n
\to \Gamma$. Regardless of the infinite volume $\Gamma$ to which the
sequence $\Lambda_n$ increases, the sequence of functionals
$\omega_{0}^{\Lambda_n}(\cdot) = \langle \psi_{0}^{\Lambda_n} ,
\cdot \,\psi_{0}^{\Lambda_n}\rangle$ will converge to the product state
$\omega_0^{\Gamma}(\cdot)$ such that with respect to its the GNS representation
$(\pi_{0}^{\Gamma}, \, \cH_{0}^{\Gamma}, \, \Omega_0^\Gamma)$, it is given by
the vector state
\begin{equation}\label{eqn:product.g.s}
\omega_0^{\Gamma}(\cdot) =
\braket{\Omega_{0}^\Gamma}{\pi_{\omega_0^\Gamma}(\cdot)\, \Omega_{0}^\Gamma}.
\end{equation}

Depending on $\Gamma$ and the value of the parameters $\lambda_k$, one of two
possible scenarios is realized for the thermodynamic limit of the one-particle
ground state $\omega_{1}^{\Lambda_n}(\cdot) = \langle
\psi_{1}^{\Lambda_n} , \cdot \,\psi_{1}^{\Lambda_n}\rangle$. Either, it also
converges to the product vacuum, in which case $\omega_0^{\Gamma}$ is the unique limiting zero-energy ground state, or it converges to a one-particle ground state  which is realized by a unit vector
$\Omega_1^\Gamma$ also in the GNS space of the product vacuum $\omega_0^\Gamma$.
The state $\Omega_1^\Gamma$ is orthogonal to $\Omega_0^\Gamma$ and
represents the particle in a bound state. We will refer to the first possibility
as {\em Scenario I} and the second as {\em Scenario II.}  An example of
Scenario I is given by $\Gamma=\bZ^d$ in which case the product vacuum is the
unique ground state in the thermodynamic limit.  An explicit expression
of $\Omega_1^\Gamma \in \cH_{0}^{\Gamma}$ in Scenario II
is given by
\begin{equation}\label{Omega_1}
\Omega_1^\Gamma = \frac{1}{\sqrt{C(\Gamma)}} \sum_{\mbf{x} \in \Gamma}
\lambda^{\mbf{x}}
\pi_{\omega_0^{\Gamma}}(\sigma_\mbf{x}^1)\Omega_{0}^{\Gamma}
\end{equation}
where $C(\Gamma) = \sum_{\mbf{x} \in \Gamma} \lambda^{2\mbf{x}}$ and
$\sigma_\mbf{x}^1$ is the first Pauli matrix.

From the expression for $\psi_1^{\Lambda}$, it is clear that the probability
amplitude for the particle in the one-particle ground state is concentrated on
the sites $\mbf x$ in $\Lambda$ where $\lambda^{\mbf{x}}$ is maximized. If
$\Lambda$ is a rectangular box and all $\lambda_k\neq 1$, then the particle will concentrate in
a corner. It is then easy to see that an example of Scenario II is realized if
all $\lambda_k\in (0,1)$, and the thermodynamic limit is taken over an
increasing sequence of hypercubes of the form $\Lambda_n =[0,n]^d$. The
thermodynamic limit of the one-particle ground states then shows a particle in a
bound state concentrated at the origin of $\Gamma=[0,\infty)^d\subset\bZ^d$.
Which scenario, I or II, that occurs for a given infinite volume $\Gamma$ and
set of model parameters is summarized by the following proposition.

\begin{prop}\label{prop:GS_Scenarios}
Let $\Gamma$ be an arbitrary infinite lattice and let $(\Lambda_n)_{n\in\bbN}$ be a sequence of increasing and absorbing finite sets converging to $\Gamma$. Then, in the weak-* topology,
\begin{enumerate}
\item (Scenario I): If $\lim_{n\to\infty}C(\Lambda_n) = +\infty$, then
$\omega^{\Lambda_n}_{1} \to \omega^\Gamma_{0}$,
\item (Scenario II): If $\lim_{n\to\infty}C(\Lambda_n) < +\infty$, then
$\Omega_1^\Gamma \in \cH_{0}^{\Gamma}$,
$\left\langle\Omega_1^\Gamma,\Omega_0^\Gamma\right\rangle = 0$ and
$\omega^{\Lambda_n}_{1} \to \omega^\Gamma_{1}$ where
\begin{equation*}
\omega^\Gamma_{1}(\cdot) = \left\langle \Omega_1^\Gamma,
\pi_{0}^{\Gamma}(\cdot)\, \Omega_1^\Gamma \right\rangle.
\end{equation*}
\end{enumerate}
\end{prop}

The {\em excitation spectrum} of the model defined on an infinite volume
$\Gamma\subset\bZ^d$ with respect to a ground state $\omega$ is defined as the
spectrum of the GNS Hamiltonian, $H_\omega$, i.e., the densely defined
self-adjoint operator on the GNS Hilbert space $\cH_\omega$, that generates the
dynamics of the infinite system and is such that $H_\omega \Omega_\omega =0$. We
say that there is a {\em spectral gap} above the ground state $\omega$ if there
exists a $\delta>0$ such that
\begin{equation*}
\spec(H_\omega)\cap(0,\delta)=\emptyset.
\end{equation*}
In this situation the {\em spectral gap}, $\gamma$, is defined by
 \begin{equation} \label{Def_Gap}
 \gamma_\omega=\sup\{\delta>0\,:\,\text{spec}(H_\omega)\cap(0,\delta)=\emptyset\}.
 \end{equation}
If such $\delta>0$ does not exists we put $\gamma=0$ and say that the ground
state has gapless excitations.

We will show that any infinite volume obtained as a limit of
increasing rectangular boxes will have a nonzero gap in the associated excitation spectrum
given that $\lambda_k \neq 1$ for all $k = 1, \, \ldots, \, d$. Specifically,
this indicates that both $\Gamma= [0, \infty)^d$ and $\Gamma = \bZ^d$ will have
gapped excitation spectrum given that no $\lambda_k =1$. Hence, a gapped
excitation spectrum can arise for models belonging to both scenarios I and II. We believe
the excitation spectrum is gapped for any model categorized by Scenario II.
However, we present other cases of Scenario I
that have a gapless excitation spectrum above the ground state even if
$\lambda_k\neq 1$, for all $k=1,\ldots,d$. For example, the
thermodynamic limit of diamond shaped $\Lambda_n\subset \bZ^2$ with $\lambda_1=\lambda_2\in (0,1]$
discussed in Section \ref{sec:UpperBound}, has gapless excitations above the
unique ground state.

In the following proposition, which we prove in sections \ref{sec:LowerBound}
and \ref{sec:UpperBound}, the key tools to studying the spectral gap are the
upper and lower bounds for the system on rectangular boxes of the form
$\Gamma_\mbf{N}=[0,N_1]\times\cdots\times [0,N_d]\subset \bZ^d$, $N_1,\ldots,N_d\geq
1$. Let $\gamma(\Gamma)$ denote the spectral gap of the system defined on $\Gamma$, namely the spectral gap of the Hamiltonian \eq{eqn:PVBSHam} for finite systems and that of the GNS Hamiltonian $H_0^\Gamma$ of the corresponding vacuum state $\omega_0^\Gamma$ in the case of an infinite system. Let $B_d$ be the $d$-dimensional unit
hypercube in $\bZ^d$.

\begin{thm}[Bounds for the Spectral Gap]\label{thm:SpectralGap}\hfill\\
\noindent
For the PVBS model with local Hamiltonians defined in \eq{eqn:PVBSHam}, and
$\lambda_k\in (0,\infty), k=1,\ldots,d$, we have the following bounds on the
spectral gap:\\
\noindent
(i) For any finite rectangular solid $\Gamma=\Gamma_\mbf{N}$, or any infinite
$\Gamma\subset\bZ^d$ that can be obtained as the limit of a sequence of
increasing rectangular solids (i.e., translations of finite volumes $\Gamma_\mbf{N}$), we have the following lower bound:
\begin{equation}
\gamma(\Gamma)\geq \frac{\gamma(B_d)}{2^d}\prod_{k=1}^d(1 - \epsilon(\lambda_k)\sqrt{2})^2 
\label{lowerbounds}\end{equation}
where
\begin{equation}
\epsilon(\lambda_k)=\left\{\begin{matrix} 
							\frac{\lambda_k}{\sqrt{1+\lambda_k^2}}  &  \lambda_k < 1  \\
							\frac{1}{\sqrt{1+\lambda_k^2}}          &  \lambda_k > 1
                        \end{matrix}\right. 
\label{eqn:epsilon}
\end{equation}
\noindent
(ii) For the model on $\bZ^d$, the gap of the GNS Hamiltonian satisfies the
upper bound
\begin{equation}
 \gamma(\bZ^d) 
\leq \sum_{k:\lambda_k\neq 1}\frac{(1 -
\lambda_k)^2}{1 + \lambda_k^2}.
\end{equation}
In particular, it gapless if $\lambda_k = 1$ for all $k$.

\noindent
\end{thm}

Note that there is a discrepancy between the upper and lower bounds of the theorem, namely the lower bound is equal to zero as soon as \emph{one} of the model parameters is equal to $1$, while the upper bound vanishes only if \emph{all} of them equal $1$. This is due to the fact that the lower bound is a finite volume calculation while the upper bound is a true infinite volume statement. While in finite volume, there is a zero-energy state in the one-particle subspace, this is no longer true in the GNS Hilbert space of the model on $\bZ^d$. Observe that the restriction of the Hamiltonian for finite or infinite rectangular solids to the one-particle sector is separable in the sense that it is a sum of terms each of which acts on one of the particle's coordinates only. Hence, it is possible in finite volume to construct a vector in the one-particle space which is a ground state of the terms corresponding to all but one coordinate directions, but such a vector does not exist in the GNS space for the model on $\bZ^d$, and any variational state in the one-particle sector there will have an energy which is at least the sum of all one-dimensional gaps. This suggests that the GNS Hamiltonian on $\bZ^d$  is gapless only if all model parameters are equal to one.

The upper bounds in this proposition are proved in the usual way by constructing
suitable variational states. For the lower bounds we use the martingale method
of \cite{nachtergaele:1996}. The constant $\gamma(B_d)$ is positive by
definition. For example, when $d=2$,
\begin{equation*}
\gamma(B_2)= 2-\sqrt{1+\frac{4\lambda_1\lambda_2}{(1+\lambda_1^2)(1+\lambda_2^2)}}\,.
\end{equation*}

The model with one species of particles we study here is equivalent to an XY
model with a particular choice of magnetic field. The same analysis applies to
the following generalization of XXZ type as long as the additional parameter
$\Delta>-1$. With the same notations as in \eqref{eqn:PVBSInteraction}, the
generalization is defined by
\begin{align} 
H_{\Lambda}(\Delta)  = &\sum_{k=1}^d\sum_{\mbf x, \mbf x +e_k}h_{\mbf x, \mbf x
+e_k}^{(k)}(\Delta) \\
 h_{\mbf{x}, \mbf{x}+ e_k}^{(k)}(\Delta)  =&
(1+\Delta)\ket{1,1}\bra{1,1} + \ket{\hat{\phi}_k}\bra{\hat{\phi}_k}, \\
= &- \frac{\lambda_k}{1+\lambda_k^2}(S^+\otimes S^- +
S^-\otimes S^+) + \Delta S^3 \otimes S^3 \nonumber \\
&   - \left(\frac{\lambda_k^2}{1+\lambda_k^2}
+\frac{\Delta}{2}\right)S^3 \otimes \idtyty  - \left(\frac{1}{1+\lambda_k^2}
+\frac{\Delta}{2}\right)\idtyty \otimes S^3 \nonumber \\ &  +
\frac{1}{2}\left(1 +\frac{\Delta}{2}\right)\idtyty \otimes \idtyty. \nonumber
\end{align}

The ground state space for this generalization is the same as the ground state
space for the one species, $d$-dimensional PVBS model as long as $\Delta > -1$.
As the martingale method is a calculation on the ground state space of a model,
it follows that the spectral gap claim from Theorem \ref{thm:SpectralGap}(i)
also holds for this more general case with the slight modification that the
constant $\gamma(B_d)$ is the spectral gap of the generalized Hamilonian
$H_{B_d}(\Delta)$.

We also show that the presence of edge states can cause the spectrum in the
thermodynamic limit to be gapless. In Section \ref{sec:UpperBound} we
consider the PVBS models on $D_\infty = \{(x,y)\,: \, y\geq -x \}\subseteq \bZ^2$
with $\lambda_1 = \lambda_2 \in (0,1)$ for which Scenario I is realized, and we
prove that the excitation spectrum of the GNS Hamiltonian is gapless. Under the
latter condition,  a particle can `bind' to one of the 45 degree boundaries
while remaining delocalized on the boundary. This gives rise to a gapless band
of edge states.

Specifically, we consider a sequence $\phi_L\in\cH_{0}^{D_\infty}$ of one-particle states in the GNS Hilbert space for which the particle is confined
to a finite diamond
$D_L \subset\Gamma$, see Figure \ref{figure:Diamond} in Section
\ref{sec:UpperBound}. We show that that the energy $E_L = \|\phi_L\|^{-1}\langle
\phi_L, \, H_0^{D_\infty} \phi_L\rangle$ is of the order
$L^{-2}$, showing that the excitation spectrum of the GNS Hamiltonian is gapless in this case.

\begin{prop} \label{prop:Diamonds}
Let $d = 2$, and $\lambda_1 = \lambda = \lambda_2$ where $\lambda \in (0,1)$.
Then the GNS Hamiltonian of the one-species PVBS model is gapless on $D_\infty = \{(x,y) \, : \, y \geq
-x \}$.
\end{prop}
 Due to the symmetry of the model, the spectrum remains unchanged if
$\lambda_k\to \lambda_k^{-1}$ and $x_k\to -x_k$. Therefore, similar statements
about the ground states and the excitation spectrum can also be made for the
model defined on suitable lattices with coupling constants in $(1,\infty)$.


\section{Ground State Space}\label{sec:GroundStateSpace}


Since we consider the one species class of models, the Hilbert space at any site
$\mbf x\in\bbZ^d$ is a copy of $\bbC^2$. Let $\Lambda$ be any finite connected subset
of $\bZ^d.$ For $0\leq N\leq \vert \Lambda\vert$, let
\begin{equation*}
\caH^N_\Lambda:=\mathrm{span}\left\{\psi_X^\Lambda : \vert X\vert = N\right\},
\end{equation*}
which we refer to as the $N$-particle subspace. Clearly, $\caH_\Lambda =
\oplus_{0\leq N\leq \vert \Lambda\vert}\caH^N_\Lambda$. The
interactions defined in \eq{eqn:PVBSInteraction} are such that
$H_\Lambda\caH^N_\Lambda \subset \caH^N_\Lambda$. Thus, the Hamiltonian
preserves particle number.

By the definition of the Hamiltonian, $H_\Lambda \xi_\emptyset^\Lambda=0$.
Since the Hamiltonian $H_\Lambda$ is a sum is a positive semi-definite
operators, $H_\Lambda\geq 0$, the ground state space $\caG_\Lambda$ is
given by its non-empty kernel
\begin{equation*}
\ker(H_\Lambda)=\hskip -8pt \bigcap_{\substack{k = 1, \ldots, d \\
\mbf{x}, \mbf{x} + e_k \in \Lambda}}\hskip -8pt \ker(h^{(k)}_{\mbf{x},\,
\mbf{x}+e_k}).
\end{equation*}
In other words, 
\begin{equation}\label{bondequations}
\psi\in \mathcal{G}_\Lambda \mbox{ iff }  h^{(k)}_{\mbf{x},\,
\mbf{x}+e_k}\psi=0 \mbox{ for all } \mbf{x}, k \mbox{ such that } \{\mbf{x},
\mbf{x} + e_k \}\subset \Lambda.
\end{equation}

\begin{proof}[Proof of Proposition \ref{prop:FiniteGSS}] 

Since the Hamiltonian preserves particle number we can just look for solutions
of \eq{bondequations} in each N-particle subspaces $\caH^N_\Lambda$ separately.

(i) $N=0$. $\caH^0_\Lambda=\{\mathbb{C}\xi_\emptyset^\Lambda\}$ is
one-dimensional and we have already seen that $\psi^\Lambda_0 =
\xi_\emptyset^\Lambda$ is a ground state.

(ii) $N=1$.
Any $\psi\in\caH^1_\Lambda$ has an orthogonal expansion of the form
\begin{equation*}
\psi = \sum_{\mbf{x}\in \Lambda}a_{\mbf{x}}\xi_{\{\mbf x\}}^\Lambda
\end{equation*}
and for this case, the equations given in \eq{bondequations} are equivalent to
the following equations for the coefficients $a_\mbf{x}$:
\begin{equation}\label{coefficientequations}
a_{\mbf{x}+e_k}=\lambda_k a_\mbf{x},\
 \mbox{ for all } \mbf{x}, k \mbox{ such that } \{\mbf{x}, \mbf{x} + e_k \}\subset \Lambda.
\end{equation}
The equations
\begin{equation}
 a_{\mbf{x}} =  \mbf{\lambda}^{\mbf{x}}
\hspace{.5 cm} \text{where} \hspace{.5 cm} \lambda^{\mbf{x}} = \prod_{k = 1}^d
\lambda_k^{x_k},
\end{equation}
are a non-zero solution to \eqref{coefficientequations}. We argue that
this is the only linearly independent solution in $\caH_\Lambda^1$. Since
$\Lambda$ is connected, for any two distinct $\mbf{x}, \mbf{y}\in\Lambda$, 
there is a path $\mbf{x} = \mbf{x}_0 \to \mbf{x}_1 \to \ldots \to \mbf{x}_n =
\mbf{y}$ contained in $\Lambda$ such that for all $i = 0, \, \ldots, \, n-1$,
$\mbf{x}_{i+1} = \mbf{x}_i + p_i e_{k_i}$ where $p_i \in \{ \pm 1 \}$ and $k_i
\in \{ 1, \, \ldots, \, d\}$. This implies that $\mbf{y-x} = (n_1, \, n_2, \,
\ldots, \, n_d)$ where $n_k = \sum_{i, \, k_i = k} p_i$ and applying
\eqref{coefficientequations} shows 
$$
 a_{\mbf{y}} =  \mbf{\lambda}^{\mbf{y-x}} a_{\mbf{x}} .
 $$
Hence, the solution of \eq{coefficientequations} is unique up to a
multiplicative constant.

(iii) $N\geq 2$.
Any $\psi\in\caH^N_\Lambda$ has an expansion of the form
\begin{equation*}
\psi = \sum_{X \subset\Lambda, \vert X\vert =N}a_{X}\xi_X^\Lambda .
\end{equation*}
Due to the first term in the general interaction defined in
\eq{eqn:PVBSInteraction}, the equations \eq{bondequations} imply that $a_X=0$
whenever $X$ contains a nearest neighbor pair. In particular, for a connected
$\Lambda$ with at least two sites, it follows that there are no non-zero
solutions with $N=\vert \Lambda\vert$. Consider any $X$ a subset of $\Lambda$
such that it contains no nearest neighbor pairs and $|X|=N<\vert\Lambda\vert$.
Fix $\mbf{x}\in X$ and let $\mathbf{x'}\in\Lambda\setminus X$ be a nearest
neighbor of $\mathbf{x}$, i.e., $\mbf{x'}=\mbf{x}+p e_k$, for some
$p\in\{-1,+1\}$ and $k\in \{1,\ldots,d\}$.
Let
$\tilde{X}$ denote the set obtained from $X$ by replacing $\mathbf{x}$ with
$\mathbf{x'}$.
The equations \eq{bondequations} then imply
\begin{equation}\label{Xeq}
a_{\tilde{X}}=\lambda_k^p a_X.
\end{equation}
Since $\vert X\vert \geq 2$ and $\Lambda$ is connected, for any pair of distinct
sites $\mbf{x}, \mbf{y}\in X$ there is a path $\mbf{x} =\mbf{x}_0 \to
\mbf{x}_1 \to \ldots \to \mbf{x}_n = \mbf{y}$ contained in $\Lambda$ such that
for all $i
= 0, \, \ldots, \, n-1$, $\mbf{x}_{i+1} = \mbf{x}_i + p_i e_{k_i}$ with $p_i
\in \{ \pm 1 \}$ and $k_i \in \{ 1, \, \ldots, \, d\}$. For any $\mbf{x}\in
X$, pick a $\mbf{y}\in X$ such that there is connecting path in $\Lambda$ of
shortest length among all paths in $\Lambda$ connecting $\mbf{x}$ to a site of
$X\setminus\{\mbf{x}\}$.
For such $\mbf{y}$, the shortest connecting path contained in $\Lambda$
satisfies $\{\mbf{x}_1,\ldots,\mbf{x}_{n-1}\}\subset\Lambda\setminus X$. By
moving $\mbf{y}$ along this path to the position $\mbf{x}_1$, which is nearest
neighbor to $\mbf{x}$ and applying \eq{Xeq} at each step, we conclude $$ a_X=
\mbf{\lambda}^{\mbf{y-x}} \lambda_{k_1}^{-p_1}a_{X'}
$$ 
with $X'$ a set containing the nearest neighbor pair $\{\mbf{x},\mbf{x}_1\}$.
As argued above, $a_{X'}=0$ and since $\lambda_k \neq 0$ for all $k$
we conclude that $a_X=0$.
Hence, there are no ground states with $N\geq 2$.
\end{proof}

In Proposition \ref{prop:FiniteGSS} we have assumed that $\lambda_k>0$ for all $k = 1,\, \ldots, \, d$, and we will continue to
make that assumption throughout the rest of the paper. Note that if some of the $\lambda_k$ vanish, there may be additional
solutions of the equations \eq{coefficientequations}. This happens, e.g., when
there are sites for which the only outgoing edges are in the coordinate
directions $k$ and the associated $\lambda_k = 0$.

We now turn to the proof of Proposition~\ref{prop:GS_Scenarios}:
\begin{proof}
Let $A\in\caA_{\mathrm{loc}}$, and $X = \mathrm{supp}(A)$. Then,
\begin{equation}\label{omega_1}
\omega_{1}^{\Lambda_n}(A) = \frac{1}{C(\Lambda_n)}\Big[\sum_{\mbf x,\mbf y \in X}\lambda^{\mbf x +\mbf y}\left\langle\xi_{\{\mbf x\}}^{\Lambda_n}, A\xi_{\{\mbf x\}}^{\Lambda_n}\right\rangle + \sum_{\mbf x\in \Lambda_n\setminus X}\lambda^{2 \mbf x}\left\langle\xi_{\emptyset}^{\Lambda_n}, A\xi_{\emptyset}^{\Lambda_n}\right\rangle \Big].
\end{equation}
Since $X$ is a finite set, both scalar products above do not depend on $n$. The
first sum is a finite constant, independent on $n$, while the second is equal to
$C(\Lambda_n\setminus X) \omega^\Gamma_{0}(A)$. As $n\to\infty$, the first term
tends to zero if and only if $C(\Lambda_n)\to\infty$. Furthermore, and again
since $X$ is finite, $C(\Lambda_n\setminus X)$ and $C(\Lambda_n)$ are either
both convergent or both divergent. In the latter case
\begin{equation*}
\lim_{n\to\infty} \frac{C(\Lambda_n\setminus X)}{C(\Lambda_n)} = 1.
\end{equation*}
Hence, if $C(\Lambda_n)\to\infty$, then $\omega_{1}^{\Lambda_n}(A)\to
\omega^\Gamma_{0}(A)$.

On the other hand, if $C(\Gamma):=\lim_{n\to\infty }C(\Lambda_n) < +\infty$,
then $\omega_{1}^{\Lambda_n}(A)$ converges, but $\lim_{n\to\infty
}\omega_{1}^{\Lambda_n}(A) \neq \omega^\Gamma_{0}(A)$. Let $\Omega_1^\Gamma$ be
the formal expression~(\ref{Omega_1}). Since
\begin{equation*}
\bigg\Vert \sum_{\mbf{x} \in \Gamma}
\lambda^{\mbf{x}}
\pi_{\omega_0^{\Gamma}}(\sigma_\mbf{x}^1)\Omega_{0}^{\Gamma} \bigg\Vert \leq \sum_{\mbf{x} \in \Gamma}\lambda^{2\mbf{x}}<\infty,
\end{equation*}
as $(\sigma_\mbf{x}^1)^*\sigma_\mbf{x}^1 = 1$, we see that $\Omega_1^\Gamma$
is a well-defined vector and $\Omega_{1}^{\Gamma}\in\caH_{\omega_0^\Gamma}$.
Moreover,
\begin{equation*}
\left\langle\Omega_{0}^{\Gamma}, \Omega_{1}^{\Gamma}\right\rangle_{\caH_{\omega_0^\Gamma}} = \sum_{\mbf x\in\Gamma}\lambda^{\mbf x}\omega_{0}^\Gamma(\sigma_{\mbf x}^1) = \sum_{\mbf x\in\Gamma}\lambda^{\mbf x}\left\langle 0 , \sigma_{\mbf x}^1 0 \right\rangle_{\caH_{\mbf x}} = 0.
\end{equation*}

It remains to prove that $\omega_{1}^{\Lambda_n}$ converges to the vector state given by $\Omega_{1}^{\Gamma}$. Since $\xi_{\{\mbf x\}}^{\Lambda_n} = \sigma_{\mbf
x}^1\xi_\emptyset^{\Lambda_n}$, we have that
\begin{equation*}
\omega_{1}^{\Lambda_n}(A) = \frac{1}{C(\Lambda_n)}\sum_{\mbf x,\mbf y \in X}\lambda^{\mbf x +\mbf y}\langle  \xi_\emptyset^{\Lambda_n}, \sigma_{\mbf x}^1 A \sigma_{\mbf y}^1 \xi_\emptyset^{\Lambda_n}\rangle + \frac{C(\Lambda_n\setminus X)}{C(\Lambda_n)}\left\langle\xi_{\emptyset}^{\Lambda_n}, A\xi_{\emptyset}^{\Lambda_n}\right\rangle.
\end{equation*}
So if $C(\Lambda_n)$ converges, then 
\begin{align*}
\lim_{n\to\infty }\omega_{1}^{\Lambda_n}(A) &= \frac{1}{C(\Gamma)}\sum_{\mbf x,\mbf y \in X}\lambda^{\mbf x +\mbf y}\omega^{\Gamma}_0(\sigma_{\mbf x}^1 A \sigma_{\mbf y}^1)
+ \frac{C(\Gamma\setminus X)}{C(\Gamma)}\omega_0^\Gamma(A) \\
&= \frac{1}{C(\Gamma)}\sum_{\mbf x,\mbf y \in \Gamma} \lambda^{\mbf x +\mbf y}
\left \langle \pi_{\omega_0^\Gamma}\left(\sigma_{\mbf x}^1\right)
\Omega_0^\Gamma ,  \pi_{\omega_0^\Gamma}\left(A\right)
\pi_{\omega_0^\Gamma}\left(\sigma_{\mbf y}^1\right) \Omega_0^\Gamma
\right\rangle \\
& = \omega_1^\Gamma(A),
\end{align*}
which concludes the proof for Scenario II.
\end{proof}
We note that the first term in~(\ref{omega_1}) corresponds to the projection
$\sum_{\mbf x\in X}\vert \xi_{\{\mbf
x\}}^{\Lambda_n}\rangle\langle\xi_{\{\mbf x\}}^{\Lambda_n}\vert$, and produces
the probability of finding the particle within $X$. Hence the theorem could be
restated as follows:
$\omega_{1}^{\Lambda_n}$ converges to the vacuum state if and
only if the probability to find the particle within any fixed finite volume
tends to $0$ as $\Lambda_n\to\Gamma$.

To conclude this section we show that in Scenario~II, any infinite volume ground
state obtained as a weak-* limit of finite volume ground states is realized as a
vector in the GNS Hilbert space that is a linear combination of the two vectors
$\Omega_0^\Gamma$ and $\Omega_1^\Gamma$. Let $\omega$ be the weak-* limit of a
sequence of finite volume ground state functionals
$\omega_{a_n,b_n}^{\Lambda_n}(\cdot) = \langle \psi_n,\, \cdot \, \psi_n
\rangle$ which are generated by normalized vectors $\psi_{n} = a_n \,
\psi_{0}^{\Lambda_n} + b_n \, \psi_1^{\Lambda_n}\in \caG_{\Lambda_n}$. We call
such $\omega$ zero-energy ground states.
Since $\{a_n\}$ and $\{b_n\}$ are bounded sequences in $\bC$, there exists some
subsequence $n_i$ such that $a_{n_i} \to a$, and $b_{n_i} \to b$ for some $a, \,
b \in \bC$. Now,
\begin{equation*}
\omega(A) = \lim_{i\to\infty} \omega_{a_{n_i}, b_{n_i}}^{\Lambda_{n_i}}(A) =
\vert a\vert ^2 \omega_0^\Gamma(A) + \vert b\vert^2\omega_1^\Gamma(A) +
\lim_{i\to\infty} \left[\overline{a_{n_i}} b_{n_i} \left\langle
\psi_{0}^{\Lambda_{n_i}}, A \psi_1^{\Lambda_{n_i}} \right\rangle + c.c.\right].
\end{equation*}
But by the argument above, 
\begin{equation*}
\lim_{i\to\infty} \left\langle \psi_{0}^{\Lambda_n}, A \psi_1^{\Lambda_n}
\right\rangle = \sum_{\mbf x\in\Gamma} \lambda^{\mbf x}\omega_0^\Gamma(A
\sigma_{\mbf x}^1) = \left\langle\Omega_0^\Gamma,\pi_{\omega_0^\Gamma}(A)
\Omega_1^\Gamma\right\rangle,
\end{equation*}
so that, altogether,
\begin{equation*}
\omega(A) = \left\langle\Omega_\omega^\Gamma,\pi_{\omega_0^\Gamma}(A) \Omega_\omega^\Gamma\right\rangle,
\end{equation*}
where $\Omega_\omega^\Gamma := a\Omega_0^\Gamma + b \Omega_1^\Gamma$. 


\section{Lower Bounds for the Spectral Gap}\label{sec:LowerBound}


\begin{subsection}{The Martingale Method}

 To prove the existence of a spectral gap in the thermodynamic limit, we appeal
 to the following theorem.
  
 \begin{thm} \label{thm:GNSSG} 
 Let $H_{\omega}$ be the GNS Hamiltonian of a zero-energy ground state $\omega$ 
 for the model defined on an subset $\Gamma \subseteq \bZ^d$, and let
 $\gamma(\Gamma)$ be the spectral gap of $H_{\omega}$ as defined in
 \eqref{Def_Gap}. Then
 \begin{equation*}
 \gamma(\Gamma) \geq \liminf_{N\geq 1}\lambda_1(N)
 \end{equation*}
 where $\lambda_1(N)$ is the smallest nonzero eigenvalue of the
 frustration-free Hamiltonians $H_{\Lambda_N}$, and $\Lambda_N$ is an
 increasing and absorbing sequence of finite volumes
 $\Lambda_N\to\Gamma$.
 \end{thm} 

The martingale method provides a means for estimating the lower bound for the
spectral gap of the Hamiltonian $H_{\Lambda_N}$, for suitable finite volumes
$\Lambda_N$. To apply Theorem \ref{thm:GNSSG} to obtain a lower bound for the
infinite volume spectral gap $\gamma(\Gamma)$, we need to find a sequence of
absorbing finite volumes $\Lambda_N$ for which the martingale method yields a
uniform lower bound for the finite volume spectral gaps. For a fixed finite
 volume $\Lambda_N$, the martingale method requires a finite sequence of volumes
 $\Lambda_n$ that increase to $\Lambda_N$, for which one can verify the
 following three conditions.
\\
\\
\noindent \textbf{Conditions for the Martingale Method.}

 The following conditions must hold for a suitable value of $l>0$. 
 
 \begin{enumerate}
   \item [(1)] There exists a constant $d_l$ for which the local Hamiltonians
   satisfy \[0\leq \sum_{n=l}^NH_{\Lambda_n\backslash\Lambda_{n-l}}\leq d_l
   H_{\Lambda_N}.\]
  
   \item [(2)]The local Hamiltonians $H_{\Lambda_n}$ have a non-trivial kernel
   $\mathcal{G}_{\Lambda_n}\subseteq\mathcal{H}_{\Lambda_n}$ and a
   nonvanishing spectral gap $\gamma_l>0$ such that:
   \[H_{\Lambda_n\backslash\Lambda_{n-l}}\geq\gamma_l(\mathbb{I}-G_{\Lambda_n\backslash\Lambda_{n-l}})\]
   for all $n\geq n_l$. We denote by $G_{\Lambda_n}$ the orthogonal projection
   onto $\mathcal{G}_{\Lambda_n}$. For $\Lambda_n\subseteq\Lambda_N$,
   $G_{\Lambda_n}$ projects onto
   $\mathcal{G}_{\Lambda_n}\otimes\mathcal{H}_{\Lambda_N\backslash\Lambda_{n}}.$
   \item [(3)]There exists a constant $\epsilon_l<\tfrac{1}{\sqrt{l}}$ and
   some $n_l$ such that for all $n_l\leq n \leq N-1$,
   \[\|G_{\Lambda_{n+1}\backslash\Lambda_{n+1-l}}E_n\|\leq \epsilon_l\] where
   $E_n=G_{\Lambda_n}-G_{\Lambda_{n+1}}.$ 
\end{enumerate}
 \begin{thm}[The Martingale Method]\label{thm:MartMeth}
 Assume that conditions (1)-(3) are satisfied for the same integer $l.$ Then for
 any $\psi\in\mathcal{H}_{\Lambda_N}$ such that $G_{\Lambda_N}\psi=0$, one has
 \begin{equation}
 \langle\psi,\,H_{\Lambda_N}\psi\rangle \geq
\frac{\gamma_{l}}{d_{l}}(1-\epsilon_l\sqrt{l})^2\|\psi\|^2.
\end{equation} 
\end{thm}

A proof of Theorem \ref{thm:MartMeth} can be found in \cite{nachtergaele:1996}.
\end{subsection}

\begin{subsection}{The Spectral Gap for PVBS Models on Rectangular Boxes}

By virtue of Theorem \ref{thm:GNSSG}, the proof of the lower bound in Theorem
\ref{thm:SpectralGap}(i) will follow from appropriate lower bounds for the
spectral gap for the model on rectangular boxes, which is the focus of this
section. Recall that in the statement of the theorem we assume $\lambda_k \in
(0, \infty)$. We additionally assume that $\lambda_k \neq 1$ for all $k$ as the
lower bound in Theorem \ref{thm:SpectralGap} is trivial if any value $\lambda_k
= 1$.

Let $\mbf{N} = (N_1, \, \ldots, N_d)$ with
$N_k\geq 1$ for all $k$. By the translation invariance of the PVBS models,
proving the lower bound \eqref{lowerbounds} on any rectangular volume of the
form
\begin{equation*}
\Gamma_\mbf{N} = \prod_{k=1}^d [0, \, N_k] \subset \bZ^d
\end{equation*}
is sufficient to prove the lower bound for any $d$-dimensional rectangular
solid. The remainder of this section will be dedicated to proving the lower
bound on the spectral gap for a PBVS model defined on a general lattice of the
form $\Gamma_{\mbf{N}}$.
 
 \begin{proof}[Proof of Theorem \ref{thm:SpectralGap} (i)]
  
  Define
 \begin{equation*}
 \Gamma_{\bfN}^{(k)} = \prod_{i=1}^{d-k}[0, \, N_i] \hskip 10pt \times
 \prod_{i=d-k+1}^d [0, \, 1].
 \end{equation*}
 and note that $\Gamma_{\bfN} = \Gamma_{\bfN}^{(0)}$. Using the
 martingale method we prove for $\epsilon(\lambda_k)$ defined as in
 \eqref{eqn:epsilon} that
 \begin{equation} \label{firstMM}
 \gamma(\Gamma_{\bfN}) \geq \frac{\gamma(\Gamma_{\bfN}^{(1)})}{2} (1 -
 \epsilon(\lambda_d)\sqrt{2})^2.
 \end{equation}
 By  relabeling the axes and permuting the model
 parameters accordingly, the result \eq{firstMM} actually proves
 \begin{equation} \label{arbitraryMM}
 \gamma(\Gamma_{\bfN}^{(k)}) \geq \frac{\gamma(\Gamma_{\bfN}^{(k+1)})}{2}
 (1-\epsilon(\lambda_{d-k})\sqrt{2})^2
 \end{equation} 
 for $0 \leq k \leq d-1$. Additionally, since $\Gamma_{\bfN}^{(d)} = B_d$, the
 unit hypercube, we will obtain the desired uniform lower bound by recursively
 applying \eq{arbitraryMM}:
 \begin{equation}\label{finitelowerbound}
 \gamma(\Gamma_{\bfN}) \geq \frac{\gamma(B_d)}{2^d} \prod_{k=1}^d
 (1-\epsilon(\lambda_k)\sqrt{2})^2.
 \end{equation}
  
 Let  $N_d>1$, since if $N_d =1 $ the bound given in $\eqref{firstMM}$ is
trivial since $\Gamma_{\bfN} = \Gamma_{\bfN}^{(1)}$ and
$\tfrac{1}{2}(1-\epsilon(\lambda_d)\sqrt{2})^2<1$.
As the sequence of finite volumes increasing to $\Gamma_{\bfN}$ we choose
 \begin{equation*}
 \Lambda_n = \prod_{i=1}^{d-1}[0, \, N_{i}] \times [0, n]
 \end{equation*}
 for $1\leq n \leq N_d$. Note that $\Lambda_{N_d} = \Gamma_{\bfN}$. We will use
 $l = 2$ to satisfy the conditions of the martingale method. For this value of
 $l$, conditions (1) and (2) are easily verified since the
 PVBS Hamiltonians have range-one interactions, are translation invariant and
 are furstration-free. Specifically, we find
 \begin{equation*}
 d_l =2, \hskip 12pt \text{and} \hskip 12pt \gamma_l =
 \gamma(\Gamma_{\bfN}^{(1)}).
 \end{equation*}
 It remains to prove condition (3). For our choice of $l$ this is equivalent to
 showing there exists an $\epsilon <\tfrac{1}{\sqrt{2}}$ such that
 \begin{equation}\label{eqn:MMCond3Equiv}
 \sup_{\psi \, \in \, \caG_{\Lambda_n} \cap \caG_{\Lambda_{n+1}}^\perp
 }\frac{\|G_{\Lamdif{n-1}}\psi\|^2}{\|\psi\|^2}<\epsilon^2.
 \end{equation}
 Let $\lambda = (\lambda_1, \, \ldots, \, \lambda_d)$ be the collection of PVBS
 model parameters, and  $\mbf{x}=(x_1,\, x_2,\, \ldots,\, x_d) \in \bZ^d$.
 In the discussion that follows, the definition of the following set and
 constants will be useful:
 \begin{align}
 T_\bfN                 & = \prod_{k=1}^{d-1}[0, \, N_k] \subseteq \bZ^{d-1} \label{TN}\\
 C(T_\bfN)            & = \prod_{k=1}^{d-1}c(\lambda_k, N_k),\hskip 12pt
 \text{where} \hskip 12pt c(\lambda_k, m) = \sum_{i=0}^m \lambda_k^{2i}. 
 \nonumber
 \end{align}
 Furthermore, for $\mbf{x}\in T_\bfN$, we define $(\mbf{x},x_d) \in \bZ^d$ to
 the be $d$-tuple obtained from appending $x_d$ to $\mbf{x}$. 

 Recall the expressions for the ground states $\psi_0^\Lambda$ and
 $\psi_1^\Lambda$ given in \eq{eqn:FiniteGSB}.
 Since the ground state space on any connected lattice consists of states with
 at most one particle, it is sufficient to only consider vectors $\psi \in
 \caG_{\Lambda_n} \cap \caG_{\Lambda_{n+1}}^\perp $ in equation
 \eqref{eqn:MMCond3Equiv} that have at most one particle on $\Lamdif{n}$, as
 vectors with more than one particle will be annihilated by $G_{\Lamdif{n-1}}$.
 The following is the general form for the vector $\psi$ of interest.
\begin{equation}\label{eqn:psi}
 \psi=b_0 \psi_1^{\Lambda_n} \otimes \psi_0 ^{\Lamdif{n}} + 
	  \sum_{\mbf{x} \in T_\bfN } \left[
	  a_{\mbf{x}} \psi_0^{\Lambda_n} \otimes
	  \xi_{\{(\mbf{x}, n+1)\}}^{ \Lamdif{n}} + b_{\mbf{x}} \psi_1^{\Lambda_n} \otimes
	 \xi_{\{(\mbf{x}, n+1)\}}^{ \Lamdif{n}}
	 \right]
\end{equation} 

 where
 \begin{equation}\label{eqn:BEmptySet}
 b_0 = - \frac{ 1 }{ \sqrt{ C(T_\bfN)c(\lambda_d,n) } } \sum_{\mbf{x}
 \in T_\bfN} \mbf{\lambda} ^{ (\mbf{x}, n+1)} a_{\mbf{x}} .
 \end{equation}
 Recall that $G_{\Lamdif{n-1}}$ is of the form
 \begin{equation*}
 G_{\Lamdif{n-1}} = | \psi_{0}^{\Lamdif{n-1}} \rangle\langle
 \psi_{0}^{\Lamdif{n-1}} | + | \psi_1^{\Lamdif{n-1}} \rangle\langle
 \psi_1^{\Lamdif{n-1}} |
 \end{equation*}
 with
 \begin{equation*}
 \psi_1^{\Lamdif{n-1}}=
 \frac{1}{ \sqrt{ (1+\lambda_d^2)C(T_\bfN) } } 
 \sum_{ \mbf{x} \in T_\bfN} \left[\mbf{\lambda}^{(\mbf{x},0)}
 \xi_{\{(\mbf{x},n)\}}^{\Lamdif{n-1}} + \mbf{\lambda}^{(\mbf{x},1)}
 \xi_{\{(\mbf{x},n+1)\}}^{\Lamdif{n-1}}\right].
 \end{equation*}
Applying $G_{\Lamdif{n-1}}$ to $\psi$ yields
 \begin{align*}
 G_{\Lamdif{n-1}}\psi = & \left( \frac{ b_0 \lambda_d^n }{ \sqrt{
    c(\lambda_d,n) (1+\lambda_d^2) }} + \sum_{ \mbf{x} \in T_\bfN } \frac{
    a_{\mbf{x}} \, \mbf{\lambda}^{(\mbf{x},1)} }{
    \sqrt{(1+\lambda_d^2) C(T_\bfN)} } \right)
    \psi_{0}^{\Lambda_{n-1}} \otimes \psi_1^{\Lamdif{n-1}} \\
 & + \left( \sum_{ \mbf{x} \in T_\bfN } \frac{ b_{\mbf{x}}
     \, \mbf{\lambda}^{(\mbf{x}, 1)} \sqrt{c(\lambda_d, n-1)} }{ \sqrt
     { c(\lambda_d, n) (1+\lambda_d^2) C(T_\bfN)} }
     \right)\psi_1^{\Lambda_{n-1}} \otimes \psi_1^{\Lamdif{n-1}} \\
 &  + \frac{ b_0 \sqrt{ c(\lambda_d, n-1) } }{ \sqrt{ c(\lambda_d, n) }
       } \psi_1^{\Lambda_{n-1}}\otimes \psi_0^{\Lamdif{n-1}}.
\end{align*}

 Replacing $b_{0}$ by the expression given in \eqref{eqn:BEmptySet} and applying
 the Cauchy-Schwarz inequality, we find that $\|G_{\Lamdif{n-1}}\psi\|^2$ is
 bounded above by
\begin{align*}
 \|G_{\Lamdif{n-1}}\psi\|^2 & \leq \frac{ \lambda_d^2 c(\lambda_d, n-1) }{
      (1+\lambda_d^2) c(\lambda_d, n) } \left( \sum_{\mbf{x} \in T_\bfN}
      \left(|a_{\mbf{x}}|^2 +|b_{\mbf{x}}|^2 \right) +    \frac{ 1 }{
      c(\lambda_d, n) C(T_\bfN)} \Big| \sum_{\mbf{x} \in T_\bfN}
      \mbf{\lambda}^{(\mbf{x}, n+1)} a_{\mbf{x}} \Big|^2 \right)      \\
&  =   \frac{ \lambda_d^2 c(\lambda_d, n-1) }{ (1+\lambda_d^2) c(\lambda_d, n)
      } \| \psi \|^2.
\end{align*}
Hence,
\begin{equation*}
\sup_{\psi \in \caG_{\Lambda_n} \cap \caG_{\Lambda_{n+1}}^\perp
  }\frac{\|G_{\Lamdif{n-1}}\psi\|^2}{\|\psi\|^2} \leq \frac{ \lambda_d^2
  c(\lambda_d, n-1) }{ (1+\lambda_d^2) c(\lambda_d, n) }.
\end{equation*}

 We must consider the cases $\lambda_d<1$ and $\lambda_d>1$ separately. When
 $\lambda_d<1$ we see that as $n\to\infty$
 \begin{equation*}
 \frac{ c(\lambda_d,n-1) }{ c(\lambda_d,n) }\longrightarrow 1.
 \end{equation*}
 So for all $n$,
 \begin{equation*}
 \frac{ \lambda_d^2 c(\lambda_d,n-1) }{ (1+\lambda_d^2) c(\lambda_d,n) } \leq
 \frac{ \lambda_d^2 }{ 1+\lambda_d^2 }.
 \end{equation*}
 In this case, since $0<\lambda_d<1$,
 $\tfrac{\lambda_d^2}{1+\lambda_d^2}<\tfrac{1}{2}$.
 In the case that $\lambda_d>1$, taking the limit $n \to \infty$ produces
 \begin{equation*}
 \frac{c(\lambda_d,n-1)}{c(\lambda_d,n)}\longrightarrow \frac{1}{\lambda_d^2}.
 \end{equation*}
 Therefore, when $\lambda_d > 1$ we have that for all $n$
 \begin{equation*}
 \frac{\lambda_d^2 c(\lambda_d,n-1)}{(1+\lambda_d^2)c(\lambda_d,n)}\leq
 \frac{1}{1+\lambda_d^2}.
 \end{equation*}
 Since $\lambda_d>1$ it follows that
 $\tfrac{1}{1+\lambda_d^2}<\tfrac{1}{2}$. Therefore, for all values of
 $\lambda_d$,
 \begin{equation*}
 \sup_{\psi \in \caG_{\Lambda_n} \cap \caG_{\Lambda_{n+1}}^\perp
 }\frac{\|G_{\Lamdif{n-1}}\psi\|^2}{\|\psi\|^2} \leq \epsilon(\lambda_d)^2 <
 \frac{1}{2}
 \end{equation*}
 where $\epsilon(\lambda_d)$ is as defined in \eqref{eqn:epsilon}.
 Thus, the third condition of the martingale method is satisfied, and by Theorem
 \ref{thm:MartMeth}
 \begin{equation*}
 \gamma(\Gamma_\bfN) \geq \frac{\gamma(\Gamma_N^{(1)})}{2} ( 1 -
 \epsilon(\lambda_d) \sqrt{2})^2,
 \end{equation*} 
 as desired. Then, appealing to equations \eqref{arbitraryMM}
 and \eqref{finitelowerbound}, and Theorem \ref{thm:GNSSG} we obtain all of the
 desired lower bounds of Theorem \ref{thm:SpectralGap}(i).
\end{proof}

\end{subsection}


\section{Upper Bounds and Closures of the Spectral
Gap}\label{sec:UpperBound}


 In this section we prove the upper bound
 for $\gamma(\bZ^d)$ given in Theorem \ref{thm:SpectralGap} as well as discuss two
 ways in which the spectral gap can close. Specifically, we show that the PVBS
 models are gappless if all $\lambda_k = 1$, and
 we prove Proposition \ref{prop:Diamonds} which shows through an example that
 the existence of  spectral gap in the thermodynamic limit is dependent on
 the boundary of the infinite lattice.  We use the variational principle in the GNS Hilbert space to prove both claims. 
 
\subsection{The Bulk Phase Transition}

We first prove part (ii) of the theorem. By Proposition~\ref{prop:GS_Scenarios}, the ground state of the PVBS model is unique for all value of the parameters $\lambda_1,\ldots\lambda_d$. Let $(\pi_0^{\bbZ^d}, \caH_0^{\bbZ^d}, \Omega_0^{\bbZ^d})$ be its GNS representation. For any finite volume $\Lambda$, and any $\mbf z = (z_i,\ldots z_d)$, with $z_i\in\bbC$, let $\tilde\varphi_\Lambda(\mbf z)\in\caH_0^{\bbZ^d}$ be the vector given by
\begin{equation*}
\tilde\varphi_\Lambda(\mbf z) = \pi_0^{\bbZ^d}(A_\Lambda(\mbf z))\Omega_0^{\bbZ^d}
\end{equation*}
where the local observable $A_\Lambda(\mbf z)\in\caA_\Lambda$ is given by
\begin{equation*}
A_\Lambda(\mbf z) = \sum_{\mathbf{x}\in\Lambda}\mbf z^{\mbf{x}} \sigma_{\mbf{x}}^1.
\end{equation*}
We first note that
\begin{equation*}
\left\langle \tilde\varphi_\Lambda(\mbf z),\Omega_0^{\bbZ^d}\right\rangle
=  \sum_{\mathbf{x}\in\Lambda} \mbf z^{\mbf{x}} \left\langle \Omega_0^{\bbZ^d}, \pi_0^{\bbZ^d}\left(\sigma_{\mbf{x}}^1\right)\Omega_0^{\bbZ^d}\right\rangle
= \sum_{\mathbf{x}\in\Lambda} \mbf z^{\mbf{x}} \left\langle 0, \sigma_{\mbf{x}}^10\right\rangle = 0
\end{equation*}
so that $\tilde\varphi_\Lambda(\mbf z)$ is a good variational vector for any $\mbf z\in\bbC^d$. Furthermore,
\begin{equation*}
\left\Vert \tilde\varphi_\Lambda(\mbf z) \right\Vert^2 = \left\Vert A_\Lambda(\mbf z) \psi_0^{\Lambda} \right\Vert^2 = \sum_{\mbf x\in\Lambda}\vert \mbf z\vert^{2\mbf x}
\end{equation*}
and we shall denote $\varphi_\Lambda(\mbf z):=A_\Lambda(\mbf z) \psi_0^{\Lambda} = \sum_{\mathbf{x}\in\Lambda}\mbf z^{{\mbf{x}}} \xi_{\mbf{\{x\}}}^{\Lambda}$. Finally,
\begin{equation*}
\left\langle \tilde\varphi_\Lambda(\mbf z), H_0^{\bbZ^d} \tilde\varphi_\Lambda(\mbf z) \right\rangle 
= \omega_0^{\bbZ^d}\left(A_\Lambda(\mbf z) \str \delta(A_\Lambda(\mbf z))\right) 
= \omega_0^{\bbZ^d}\left(A_\Lambda(\mbf z) \str [H_{\Lambda^{(1)}}, A_\Lambda(\mbf z)] \right)
= \left\langle \varphi_\Lambda(\mbf z), H_{\Lambda^{(1)}} \varphi_\Lambda(\mbf z)\right\rangle
\end{equation*}
where we denoted $\Lambda^{(1)} :=\{x\in\bbZ^d: d(x,\Lambda)\leq 1\}$ and used that $H_{\Lambda^{(1)}} \psi_0^{\Lambda^{(1)}} = 0$.

For any $\mbf x\in \Lambda$ and $1\leq k\leq d$ such that $\mathbf{x},
\mbf{x}+e_k \in \Lambda$ (i.e. the bulk terms) we compute
\begin{equation*}
h_{\mbf{x}, \mbf{x}+e_k}^{(k)} \varphi_\Lambda(\mbf z) = \frac{\mbf z^\mbf x(z_k - \lambda_k)}{1+\lambda_k^2}\left(\xi_{\mbf{\{x+e_k\}}}^{\Lambda} - \lambda_k\xi_{\mbf{\{x\}}}^{\Lambda}\right)
\end{equation*}
so that
\begin{equation*}
\left\langle \varphi_\Lambda(\mbf z), H_{\Lambda}\varphi_\Lambda(\mbf z)\right\rangle = \sum_{\mbf{x}\in\Lambda}\vert \mbf z \vert^{2\mbf x}\sum_{k=1}^d\frac{(z_k - \lambda_k)(\bar z_k - \lambda_k)}{1+\lambda_k^2}.
\end{equation*}
For the boundary terms, we have for any $1\leq k\leq d$,
\begin{equation*}
\left\langle {\varphi_\Lambda}(\mbf z),h_{\mbf{x}, \mbf{x}+e_k}^{(k)} \varphi_\Lambda(\mbf z)\right\rangle = 
\begin{cases}
\frac{\lambda_k^2}{1+\lambda_k^2}\vert \mbf z \vert^{2\mbf x} & \text{if }\mbf x\in\Lambda, \mbf x + e_k\notin\Lambda \\
\frac{\vert z_k \vert^2}{1+\lambda_k^2}\vert \mbf z \vert^{2\mbf x} & \text{if }\mbf x\notin\Lambda, \mbf x + e_k\in\Lambda
\end{cases}
\end{equation*}

We now choose for convenience the finite set to be
\begin{equation*}
\Lambda=\bar\Gamma_\bfN = \prod_{k=1}^d [-N_k, \, N_k] \subset \bZ^d.
\end{equation*}
with $\left\vert \bar\Gamma_\bfN\right\vert = \prod_{k=1}^d(2N_k + 1)$. The sum of all boundary terms reads
\begin{equation*}
\sum_{k=1}^d\frac{1}{1+\lambda_k^2} \left(\lambda_k^2\vert z_k\vert^{2N_k} + \vert z_k\vert^{-2(N_k-1)}\right) \prod_{i\neq k}\sum_{\alpha = -N_i}^{N_i} \vert z_i \vert^{2\alpha}.
\end{equation*}

Let us consider the special case $z_k = 1$ for all $k$. Then $\left\Vert \tilde\varphi_\Lambda(\mbf z) \right\Vert^2 = \left\vert \bar\Gamma_\bfN\right\vert$ and the boundary contributions read
\begin{equation*}
\left\vert \bar\Gamma_\bfN\right\vert\sum_{k=1}^d\frac{1}{2N_k + 1}\, .
\end{equation*}
Both of these expressions are independent of $\lambda$. Furthermore, the bulk terms are bounded by
\begin{equation*}
\left\vert \bar\Gamma_\bfN\right\vert \sum_{k: \lambda_k\neq 1}\frac{(1-\lambda_k)^2}{1+\lambda_k^2},
\end{equation*}
so that
\begin{equation*}
\frac{\left\langle \tilde\varphi_{\bar\Gamma_\bfN}(1), H_0^{\bbZ^d} \tilde\varphi_{\bar\Gamma_\bfN}(1) \right\rangle}{\left\Vert \tilde\varphi_{\bar\Gamma_\bfN}(1)\right\Vert^2} \leq \sum_{k: \lambda_k\neq 1}\frac{(1-\lambda_k)^2}{1+\lambda_k^2} + \sum_{k=1}^d\frac{1}{2N_k + 1}.
\end{equation*}
which yields the estimate
\begin{equation*}
\gamma(\bbZ^d)\leq\liminf_{\bfN\to\infty} \frac{\left\langle \tilde\varphi_{\bar\Gamma_\bfN}(1), H_0^{\bbZ^d} \tilde\varphi_{\bar\Gamma_\bfN}(1) \right\rangle}{\left\Vert \tilde\varphi_{\bar\Gamma_\bfN}(1)\right\Vert^2}  \leq \sum_{k: \lambda_k\neq 1}\frac{(1-\lambda_k)^2}{1+\lambda_k^2}
\end{equation*}
In particular, the model is gapless if $\lambda_k = 1$ for all $1\leq k\leq d$.

\subsection{Low-lying Edge States} 
In this section we explore how the thermodynamic behavior depends on the
boundary of the infinite lattice. We previously showed that the existence of a
unique GNS ground state vector is closely tied with the geometry, and more
specifically the boundary, of the infinite volume one considers. It is then
natural to ask whether the geometry will effect other properties in the
thermodynamic limit. 

As an example, we discuss the case of the diagonal-edged infinite volume
$D_\infty\subset \bZ^2$ defined by $(x,y) \in D_\infty$ if and only if $y \geq -x$
and prove that, unlike an infinite lattice obtained as the limit of finite
rectangles, the PVBS model with one species of particle is in fact gapless when
$\lambda_1=\lambda=\lambda_2$ with $\lambda < 1$. The closure of the gap in this
case is due to the existence of boundary states in which a single particle
is delocalized across the boundary.

\begin{proof}[Proof of Proposition~\ref{prop:Diamonds}]
We construct a sequence of variational vectors in the GNS
Hilbert space for which the energy decays polynomially. The vectors we choose
belong to the one particle sector with the particles bound to the finite
volume $D_L$ defined by:
 \bea
 D_L=\{(x,y): 0\leq x+y\leq L,|x-y|\leq L/2\}.
\eea

\begin{figure}[htbp]
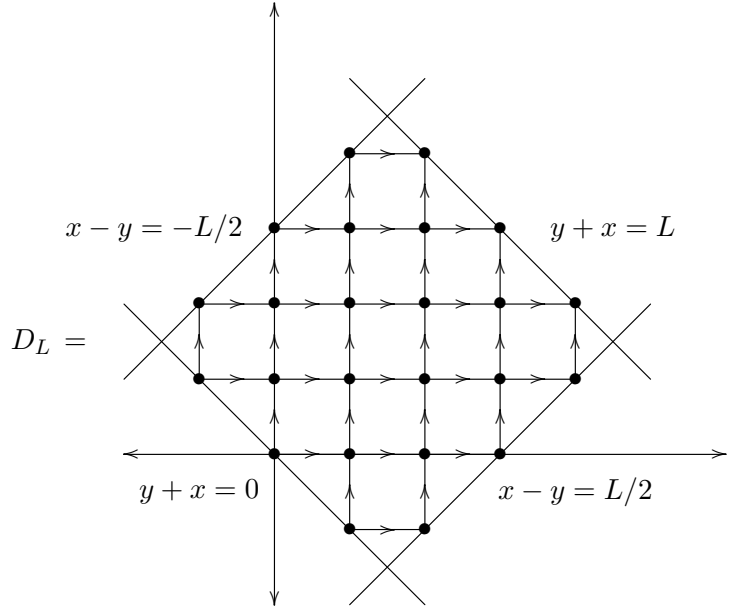

 \begin{center}
\[
\xy
(-30,15)*+{D_L\,=};
(10,-10)*+{\bullet}; 
(20,-10)*+{\bullet};
(0,0)*+{\bullet}; 
(10,0)*+{\bullet}; 
(20,0)*+{\bullet}; 
(30,0)*+{\bullet}; 
(-10,10)*+{\bullet}; 
(0,10)*+{\bullet}; 
(10,10)*+{\bullet}; 
(20,10)*+{\bullet}; 
(30,10)*+{\bullet}; 
(40,10)*+{\bullet}; 
(-10,20)*+{\bullet};
(0,20)*+{\bullet}; 
(10,20)*+{\bullet}; 
(20,20)*+{\bullet}; 
(30,20)*+{\bullet}; 
(40,20)*+{\bullet}; 
(0,30)*+{\bullet}; 
(10,30)*+{\bullet}; 
(20,30)*+{\bullet}; 
(30,30)*+{\bullet}; 
(10,40)*+{\bullet}; 
(20,40)*+{\bullet};
(20,-20); (-20,20) **\dir{-};
(-10,-5)*+{y+x=0};
(-20,10); (20,50)**\dir{-};
(-16,30)*+{x-y=-L/2};
(50,10); (10,50)**\dir{-};
(45,30)*+{y+x=L};
(50,20); (10,-20)**\dir{-};
(40,-5)*+{x-y=L/2};
(-10,10); (-10,20) **\dir{-};
(0,-20);  (0,60) **\dir{-};
(10,-10);  (10,40) **\dir{-};
(20,-10);  (20,40) **\dir{-};
(30,0);  (30,30) **\dir{-};
(40,10);  (40,20) **\dir{-};
(10,-10); (20,-10) **\dir{-};
(-20,0);  (60,0) **\dir{-};
(0,0);  (30,0) **\dir{-};
(-10,10);  (40,10) **\dir{-};
(-10,20);  (40,20) **\dir{-};
(0,30);  (30,30) **\dir{-};
(10,40); (20,40) **\dir{-};
{\ar@{->} (-4.1,10)*{}; (-4.09,10)*{}};
{\ar@{->} (-4.1,20)*{}; (-4.09,20)*{}};
{\ar@{->} (15.9,-10)*{}; (15.91,-10)*{}};
{\ar@{->} (15.9,40)*{}; (15.91,40)*{}};
{\ar@{->} (5.9,0)*{}; (5.91,0)*{}};
{\ar@{->} (15.9,0)*{}; (15.91,0)*{}};
{\ar@{->} (25.9,0)*{}; (25.91,0)*{}};
{\ar@{->} (5.9,10)*{}; (5.91,10)*{}};
{\ar@{->} (15.9,10)*{}; (15.91,10)*{}};
{\ar@{->} (25.9,10)*{}; (25.91,10)*{}};
{\ar@{->} (35.9,10)*{}; (35.91,10)*{}};
{\ar@{->} (5.9,20)*{}; (5.91,20)*{}};
{\ar@{->} (15.9,20)*{}; (15.91,20)*{}};
{\ar@{->} (25.9,20)*{}; (25.91,20)*{}};
{\ar@{->} (35.9,20)*{}; (35.91,20)*{}};
{\ar@{->} (5.9,30)*{}; (5.91,30)*{}};
{\ar@{->} (15.9,30)*{}; (15.91,30)*{}};
{\ar@{->} (25.9,30)*{}; (25.91,30)*{}};
{\ar@{->} (-10,15.9)*{}; (-10,15.91)*{}};
{\ar@{->} (0,5.9)*{}; (0,5.91)*{}};
{\ar@{->} (0,15.9)*{}; (0,15.91)*{}};
{\ar@{->} (0,25.9)*{}; (0,25.91)*{}};
{\ar@{->} (10,-4.1)*{}; (10,-4.09)*{}};
{\ar@{->} (10,5.9)*{}; (10,5.91)*{}};
{\ar@{->} (10,15.9)*{}; (10,15.91)*{}};
{\ar@{->} (10,25.9)*{}; (10,25.91)*{}};
{\ar@{->} (10,35.9)*{}; (10,35.91)*{}};
{\ar@{->} (20,-4.1)*{}; (20,-4.09)*{}};
{\ar@{->} (20,5.9)*{}; (20,5.91)*{}};
{\ar@{->} (20,15.9)*{}; (20,15.91)*{}};
{\ar@{->} (20,25.9)*{}; (20,25.91)*{}};
{\ar@{->} (20,35.9)*{}; (20,35.91)*{}};
{\ar@{->} (30,5.9)*{}; (30,5.91)*{}};
{\ar@{->} (30,15.9)*{}; (30,15.91)*{}};
{\ar@{->} (30,25.9)*{}; (30,25.91)*{}};
{\ar@{->} (40,15.9)*{}; (40,15.91)*{}};
{\ar@{->} (0,59.99)*{}; (0,60)*{}};
{\ar@{->} (0,-19.99)*{}; (0,-20)*{}};
{\ar@{->} (59.99,0)*{}; (60,0)*{}};
{\ar@{->} (-19.99,0)*{}; (-20,0)*{}};
\endxy 
\]
   \end{center}
   \caption{\footnotesize The sector $D_L$ for $L=6$.} 
   \label{figure:Diamond}
\end{figure}
For convenience, we choose $L = 2k$ where $k$ is an odd integer. An example of
the sector $D_L$ to which to particles are bound is given in Figure
\ref{figure:Diamond}. We partition $D_L$ into the following subsets:
\begin{align*}
 D_L^{\mathrm{int}}  &=\{(x,y): 0< x+y < L,|x-y|< L/2\}
\\D_L^{\mathrm{edge}}&=\{(x,y): x+y=0\}
\\D_L^{\mathrm{opp}}&=\{(x,y):x+y=L\}
\\D_L^{\mathrm{uside}}&=\{(x,y):x-y=-L/2\}
\\D_L^{\mathrm{lside}}&=\{(x,y):x-y=L/2\}
\end{align*}
We consider the GNS Hamiltonian $H^{D_\infty}_0$  associated with the
infinite volume ground state given in~(\ref{eqn:product.g.s}),
where $\Omega^{D_\infty}_0 = \bigotimes_{(x,y)\in D_\infty}|0\rangle$. The
variational vectors $\tilde{\varphi}_L\in \mathcal{H}^{D_\infty}_0$ are defined by
\begin{equation}
\tilde{\varphi}_L = \pi^{D_\infty}_0(A_L) \Omega^{D_\infty}_0
\end{equation}
where $A_L\in\caA_{D_L}$ is the local observable
\begin{equation*}
A_L = \sum_{\mbf{x} = (x,y)\in D_L}\lambda^{x+y}\sin
(k(x-y)) \sigma_{\mbf{x}}^1
\end{equation*}
with $k = \tfrac{2\pi}{L}$. Again, $\Vert \tilde{\varphi}_L \Vert = \Vert A_L \psi_0^{D_{L}^\partial}\Vert$ and 
\begin{equation*}
\left\langle \tilde{\varphi}_L, H^{D_\infty}_0 \tilde{\varphi}_L \right\rangle 
= \left\langle A_L \psi_0^{D_{L}^\partial}, H_{D_{L}^\partial} A_L \psi_0^{D_{L}^\partial}\right\rangle
\end{equation*}
where
\begin{equation*}
D_{L}^\partial =\{(x,y): 0\leq x+y\leq L + 1,|x-y|\leq L/2 + 1\}.
\end{equation*}
and we used that $H_{D_{L}^\partial} \psi_0^{D_{L}^\partial} = 0$. Notice that $(x,y)\in D_L^{\mathrm{edge}}$ if and only if $(x+\tfrac{L}{2},
y+\tfrac{L}{2})\in D_L^{\mathrm{opp}}$, and so
\begin{equation}
\sum_{(x,y)\in D_L^{\mathrm{opp}}} \sin^2 k(x-y) \hskip 10pt = \sum_{(x,y)\in D_L^{\mathrm{edge}}} \sin^2
k(x-y).
\label{eqn:SineSquared}
\end{equation}
Using \eqref{eqn:SineSquared} and our choice of $k$,
\begin{align}
\langle A_L \psi_0^{D_{L}^\partial}, H_{D_{L}^\partial} A_L \psi_0^{D_{L}^\partial}\rangle
&= 2(1-\cos(k))
\bigg[\sum_{(x,y)\in D_L^{\mathrm{int}}}\hskip -10pt \lambda^{2(x+y)}\sin^2 k(x-y)+
\frac{\lambda^2+ \lambda^{2L}}{1+\lambda^2} \hskip -5pt \sum_{(x,y)\in
D_L^{\mathrm{opp}}} \hskip -5pt \sin^2 k(x-y)\bigg] \\
&\quad + \frac{2\lambda^{2L+2}}{1+\lambda^2} \sum_{(x,y)\in
D_L^{\mathrm{opp}}}\hskip -5pt\sin^2 k(x-y)  \nonumber \label{eqn:HamInProd}
\end{align}
and
\begin{equation*}
\left\Vert A_L \psi_0^{D_{L}^\partial}\right\Vert^2= \sum_{(x,y)\in D_L^{\mathrm{int}}}\lambda^{2(x+y)}\sin^2
k(x-y)+(1 +\lambda^{2L})\sum_{(x,y)\in D_L^{\mathrm{opp}}} \sin^2 k(x-y).
\end{equation*}
All in all, it follows that the spectral gap is
bounded above by 
\begin{equation*}
\gamma(D_\infty) \leq \dfrac{\left\langle \tilde{\varphi}_L, H^{D_\infty}_0
\tilde{\varphi}_L \right\rangle}{\|\tilde{\varphi}_L\|^2}\leq 2\left(1-\cos\left(2\pi/L\right)\right) + \frac{2\lambda^{2L+2}}{1+\lambda^2},
\end{equation*}
where we used that $\sum_{(x,y)\in D_L^{\mathrm{opp}}} \sin^2 k(x-y) \leq
\|\tilde{\varphi}_L\|^2$. Since $\lambda<1$, the second term decays exponential in $L$ as $L\to\infty$, while the first one is of order $L^{-2}$.
\end{proof}

Analogously, we expect that this type of gap closure to occur for any
$d$-dimensional infinite volume with slant boundary when the outward
pointing normal is given by the vector $\vec{n} = (\log \lambda_1, \ldots, \,
\log \lambda_d)$.


\section{Discussion}\label{sec:Discussion}


\subsection{Stability}

\emph{Stability} refers to the continuous behavior of certain spectral
properties of quantum spin Hamiltonians under small local perturbations.
Although the case of perturbations around classical Hamiltonians was considered
in great detail long
ago~\cite{Ginibre:1969,Kirkwood:1983,Matsui:1990,Kennedy:1992,Borgs:1996,Datta:1996},
results for perturbations of frustration-free quantum models with interesting
ground states have only been achieved rather recently~\cite{Yarotsky:2006,
Bravyi:2010, Bravyi:2011, Michalakis:2013,Cirac:2013,Szehr:2014}. The most basic
stability results establish two properties; for models with a unique ground
state in the infinite volume limit and for sufficiently small perturbations, one
proves first of all that the spectral gap above the ground state energy does not
close, and secondly, that the unique ground state depends continuously on the
perturbation.
In~\cite{Michalakis:2013}, such stability results are obtained for models with a
unique frustration-free ground state under two additional conditions.
First, the perturbation is of the form $\sum_{X\subset\Lambda}\Psi(X)$ where,
roughly speaking, $\Vert \Psi(X) \Vert$ decays faster than any polynomial as a
function of the diameter of $X$. Second, the possibly many finite volume ground
states satisfy what is referred to as the \emph{local topological order}
condition in \cite{Michalakis:2013}. For any hypercube
$\Lambda_L$ of length $L$, consider sets $X\subset\Lambda_L$ such that
$X^{(l)}\subset \Lambda_L$, where $X^{(l)} = \{x\in\Gamma:d(x,X)\leq l\}$. In
the case of a multi-dimensional finite volume ground state space, the ground
states are called topologically ordered if there exists a positive function $f$
decaying faster than any polynomial such that for any
$A\in\caA(X)$,
\begin{equation}\label{TQO}
\left\Vert G_{X^{(l)}} A G_{X^{(l)}} - c_l(A) G_{X^{(l)}}\right\Vert \leq \Vert
A \Vert f(l)
\end{equation}
where
\begin{equation*}
 c_l(A) = \Tr(G_{X^{(l)}}) ^{-1}\Tr(G_{X^{(l)}} A).
\end{equation*}

We now prove that the models considered here satisfy \eq{TQO}. The spaces
$G_{X^{(l)}}$ are two-dimensional and spanned by $\psi_0^{X^{(l)}}$ and
$\psi_1^{X^{(l)}}$. We shall use the bound
\begin{equation*}
\left\Vert G_{X^{(l)}} A G_{X^{(l)}} - c_l(A) G_{X^{(l)}}\right\Vert 
\leq 2 \sup_{i,j\in\{0,1\}} \left\vert\left\langle \psi_i^{X^{(l)}} , A \psi_j^{X^{(l)}}\right\rangle - c_l(A) \delta_{ij}\right\vert.
\end{equation*}
Since $A$ is supported in $X$, we have that
\begin{equation*}
\left\langle \psi_1^{X^{(l)}} , A \psi_1^{X^{(l)}}\right\rangle = \frac{C(X)}{C(X^{(l)})} \left\langle \psi_1^{X} , A \psi_1^{X}\right\rangle + \frac{C(X^{(l)}\setminus X)}{C(X^{(l)})} \left\langle \psi_0 ^{X^{(l)}\setminus X} , A \psi_0^{X^{(l)}\setminus X}\right\rangle,
\end{equation*}
and the last scalar product can be replaced by $\left\langle\psi_0 ^{ X}
, A \psi_0^{ X}\right\rangle $. With this, both cases $i=j=0$
and $i=j=1$ are bounded above by
\begin{equation*}
\frac{1}{2} \left\vert \left\langle \psi_1^{X^{(l)}} , A \psi_1^{X^{(l)}}\right\rangle - \left\langle \psi_0^{X^{(l)}} , A \psi_0^{X^{(l)}}\right\rangle \right\vert
\leq \frac{\Vert A \Vert }{2} \left[\left\vert \frac{C(X)}{C(X^{(l)})} \right\vert 
+ \left\vert1-\frac{C(X^{(l)}\setminus X)}{C(X^{(l)})}\right\vert  \right] = \Vert A \Vert\left\vert \frac{C(X)}{C(X^{(l)})} \right\vert 
\end{equation*}
where we used that $C(X^{(l)}) - C(X^{(l)}\setminus X) = C(X)$ by the
definition~(\ref{CLambda}) of $C(\cdot)$. The off-diagonal case can be computed
similarly,
\begin{equation*}
\left\vert \left\langle \psi_0^{X^{(l)}} , A \psi_1^{X^{(l)}}\right\rangle\right\vert = \left\vert \left(\frac{C(X)}{C(X^{(l)})} \right)^{1/2}\left\langle \psi_0^{X} , A \psi_1^{X}\right\rangle\right\vert 
\leq \Vert A \Vert \left\vert \frac{C(X)}{C(X^{(l)})} \right\vert ^{1/2}.
\end{equation*}
Hence, \eqref{TQO} holds with $f(l) = 2\sqrt{C(X)/C(X^{(l)})}$, which decays
exponentially in $l$.

The models considered here have open boundary conditions, while the theorem
of~\cite{Michalakis:2013} assumes periodic boundary conditions. We believe,
however, that stability holds in the present more general setting, given the
frustration-free and topological quantum order conditions. This means that there
is an open neighborhood of local interactions around the models presented in
this work which have a spectral gap above the ground state energy, on $\bbZ^d$.

\subsection{Edge versus Bulk Gapless Excitations}

To conclude, we would like to underline the physical difference between the
gapless spectrum that arises in the case of $\lambda_k = 1$ for all $k$, and
the gapless spectrum we found in Section \ref{sec:UpperBound} for a class of
models on the slanted half-infinite plane $D_\infty$. As explained above, the
limit $\lambda\to1$ corresponds to a quantum phase transition in the sense
that the bulk dispersion relation becomes gapless. In the case of the slanted
edge the low-lying excited states are localized along the edge and nothing
special happens in the bulk.
In particular, excited states localized away from the edge show a gapped energy
spectrum, and the dynamics of observables whose support is at a distance $d$
from the edge is unaffected by the gapless excitations up to times $t$ that are
of the same order of magnitude as $d/v$, where the constant $v$ is the
Lieb-Robinson velocity of the model.

To make these statements more precise, consider the PVBS model on
$D_\infty\subset \bbZ^2$, the half-space with a 45 degree edges as discussed in
Section \ref{sec:UpperBound}, and let $\omega_0^{D_\infty}$ be the
unique vacuum state of this system with $(\pi_0^{D_\infty},\caH_0^{D_\infty},
\Omega_0^{D_\infty})$ its GNS representation. Let
$D_\infty^{\mathrm{edge}}\subset D_\infty$ denote the sites at the 45 degree
boundary of $D_\infty$.

Let
\begin{equation*}
\caH_0^{\mathrm{bulk}}:=\mathrm{span}\{\pi(A)\Omega_0^{D_\infty}: d(\mathrm{supp}(A), D_\infty^{\mathrm{edge}})\geq 1\}\subset\caH_0^{D_\infty},
\end{equation*}
and $P_0^{\mathrm{bulk}}$ the orthogonal projection onto
$\caH_0^{\mathrm{bulk}}$. Finally, let $H_{0}^{D_\infty}$ be the GNS Hamiltonian
associated with $\omega_o^{D_\infty}$.
Although the GNS Hamiltonian $H_0^{D_{\infty}}$ for the model on $D_\infty$ does
not have a gap above the ground state, the self-adjoint operator $P_0
H_0^{D_{\infty}}P_0$ regarded as an operator on $\cH_0^{D_{\infty}}$ does have a
spectral gap above its ground state at least as large as the gap of the model
defined on $\bbZ^2$.
\begin{prop}
With the notations above, let
\begin{equation*}
H_0^{\mathrm{bulk}} :=
P_0^{\mathrm{bulk}}H_{0}^{D_\infty}P_0^{\mathrm{bulk}}:
\caH_0^{\mathrm{bulk}}\to \caH_0^{\mathrm{bulk}}.
\end{equation*}
Then, $\inf\mathrm{Spec}(H_0^{\mathrm{bulk}}) = 0$ is an eigenvalue and if $\lambda\neq 1$, there is
a spectral gap $\gamma\geq \gamma(\bbZ^2)$.
\end{prop}
\begin{proof}
First, we observe that $P_0^{\mathrm{bulk}}\Omega_0^{D_\infty} =
\Omega_0^{D_\infty}$ so that
\begin{equation*}
H_0^{\mathrm{bulk}} \Omega_0^{D_\infty} =
P_0^{\mathrm{bulk}}H_{0}^{D_\infty}\Omega_0^{D_\infty} = 0
\end{equation*}
which shows that $\Omega_0^{D_\infty}$ belongs to the kernel of the positive
operator $H_0^{\mathrm{bulk}}$.

We consider $A\in\cA_X$, with $d(X,D_\infty^{\mathrm{edge}})\geq 1$, and observe that the
generator of the dynamics on $D_\infty$, namely $\delta_{D_\infty}$, satisfies
\begin{equation}\label{deltas}
\delta_{D_\infty}(A)=\lim_{\Lambda\uparrow D_\infty} [H_\Lambda,A]=\lim_{\Lambda\uparrow \bbZ^2} [H_\Lambda,A] = \delta_{\bbZ^2}(A).
\end{equation}
Hence,
\begin{align*}
\langle \pi_0^{D_\infty}(A) \Omega_0^{D_\infty}, H_0^{\mathrm{bulk}} \pi_0^{D_\infty}(A)\Omega_0^{D_\infty}\rangle
&= \langle \pi_0^{D_\infty}(A) \Omega_0^{D_\infty}, H_{0}^{D_\infty}  \pi_0^{D_\infty}(A)\Omega_0^{D_\infty}\rangle
= \omega_0^{D_\infty}(A^* \delta_{D_\infty}(A)) \\
&= \omega_0^{\bbZ^2}(A^* \delta_{\bbZ^2}(A)),
\end{align*}
where in the last equality we first use ~(\ref{deltas}) and then the fact that
both ground states are the same product state on local observables. Now let $A$
be such that $\pi_0^{D_\infty}(A) \Omega_0^{D_\infty} \perp
\Omega_0^{D_\infty}$, or equivalently $0 = \omega_0^{D_\infty}(A) =
\omega_0^{\bbZ^2}(A)$.
Since the model on $\bbZ^2$ has a gap $\gamma(\bbZ^2)>0$, this implies
\begin{equation}\label{Pgap}
\langle \pi_0^{D_\infty}(A) \Omega_0^{D_\infty}, H_0^{\mathrm{bulk}} \pi_0^{D_\infty}(A)\Omega_0^{D_\infty}\rangle
\geq \gamma(\bbZ^2) \omega_0^{D_\infty}(A^*A)= \gamma(\bbZ^2) \Vert \pi_0^{D_\infty}(A)\Omega_0^{D_\infty}\Vert^2.
\end{equation}
The inequality \eq{Pgap} shows that $H_0^{\mathrm{bulk}}$ has a spectral gap above $0$ at least equal to
$\gamma(\bbZ^2)$, which we have shown to be strictly positive.
\end{proof}

Finally, we show that the dynamics of observables supported away from the edge
is unaffected by the gapless edge modes. In particular, multi-time correlation
functions for such observables behave as if the system was gapped for times that
are not loo large. To this end, recall that a nearest-neighbor interaction
defined on a regular lattice -- such as the present PVBS models -- satisfies a
Lieb-Robinson bound, with velocity $v$. In particular, the difference of the
dynamics generated by two finite-volume Hamiltonians $H_{\Lambda_1}$ and
$H_{\Lambda_2}$, where $\Lambda_2\subset\Lambda_1$, can be estimated
by 
\begin{equation*}
\left\| \tau_t^{\Lambda_1}(A) - \tau_t^{\Lambda_2}(A) \right\|  \leq
C\, \| A \|  \,
| X | e^{-a[d(X,\Lambda_1\setminus\Lambda_2) -v\vert t\vert]},
\end{equation*}
where $A\in \cA_X$, $X\subset\Lambda_2$, and with suitable constants $a, C>0$
(see ~\cite[equation (2.28)]{nachtergaele:2006}).
Let $\Lambda_{1,n}\to \bbZ^2$ and $\Lambda_{2,n} \to D_\infty$ be increasing
sequences such that $\Lambda_{2,n}\subset\Lambda_{1,n}$ for all $n\in\bbN$. For
$n$ large enough, $d(X,\Lambda_{1,n}\setminus\Lambda_{2,n}) =
d(X,D_\infty^{\mathrm{edge}})$, and in particular
\begin{equation*}
\left\| \tau_t^{D_\infty}(A) - \tau_t^{\bbZ^2}(A) \right\|  \leq
C\, \| A \|  \,
| X | e^{-a[d(X,D_\infty^{\mathrm{edge}}) -v\vert t\vert]}.
\end{equation*}
Hence, for times $t\ll d(X,D_\infty^{\mathrm{edge}}) / v$, the half-plane
dynamics of $A$ are well-approximated by the dynamics in the bulk generated by
the gapped Hamiltonian.

\section*{Acknowledgements}

This research was supported in part by the National Science Foundation: S.B.
under Grant \#DMS-0757581 and B.N. and A.Y.
under grant \#DMS-1009502. A.Y. also acknowledges support from the National
Science Foundation under Grant \#DMS-0636297. A.Y. would also like to thank the
Mathematisches Institut at Ludwig-Maximilians-Universit\"at M\"unchen and the
Erwin Schr\"odinger International Institute for Mathematical Physics (ESI) in
Vienna, Austria, where part of the work reported here was carried out. E. H  gratefully acknowledges the support of a Fulbright scholarship spent at University of California, Davis. 

\providecommand{\bysame}{\leavevmode\hbox to3em{\hrulefill}\thinspace}

\bibliographystyle{plain}

\end{document}